\documentclass[journal, twocolumn, 10pt]{IEEEtran}

\usepackage{graphicx}
\usepackage{cite}
\usepackage{amssymb,amsmath}
\usepackage{amsthm}
\usepackage{algorithm}
\usepackage{algpseudocode}
\usepackage{color}
\usepackage{epstopdf}
\usepackage{subcaption}
\newtheorem{theorem}{Theorem}

\newtheorem{proposition}{Proposition}
\newtheorem{remark}{Remark}

\newtheorem{definition}{Definition}

\begin{document}
\title{\huge{ Improper Signaling in Two-Path Relay Channels }}
\author{{ Mohamed Gaafar\IEEEauthorrefmark{2}, Osama Amin\IEEEauthorrefmark{3}, Rafael F. Schaefer\IEEEauthorrefmark{2}, and Mohamed-Slim Alouini\IEEEauthorrefmark{3}} \\[0.25cm]
\IEEEauthorrefmark{2} Information Theory and Applications Chair, Technische Universit\"at Berlin, Germany\\
\IEEEauthorblockA{\IEEEauthorrefmark{3} Computer, Electrical, and Mathematical Sciences and Engineering (CEMSE) Division\\
King Abdullah University of Science and Technology (KAUST), Thuwal, Makkah Province, Saudi Arabia.\\[0.25cm]
E-mail: {\{{mohamed.gaafar, rafael.schaefer\}@tu-berlin.de}}, {\{{osama.amin, slim.alouini\}@kaust.edu.sa}}\\[0.25cm]
(\emph{Extended Version})
}}
\maketitle
%\pagenumbering{gobble}
\begin{abstract}
Inter-relay interference (IRI) challenges the operation of two-path relaying systems. Furthermore, the unavailability of the channel state information (CSI) at the source and the limited detection capabilities at the relays prevent neither eliminating the interference nor adopting joint detection at the relays nodes. Improper signaling is a powerful signaling scheme that has the capability to reduce the interference  impact at the receiver side and improves the achievable rate performance. Therefore, improper signaling is adopted  at both relays, which have access to the global CSI. Then,  improper signal characteristics  are designed  to maximize the total end-to-end achievable rate at the relays. To this end, both the power and the circularity coefficient, a measure of the impropriety degree of the signal,  are optimized at the relays. Although the optimization problem is not convex, optimal power allocation for both relays for a fixed circularity coefficient is obtained.  Moreover, the circularity  coefficient is tuned to maximize the rate for a given power allocation. Finally, a joint solution of the optimization problem is proposed using a coordinate descent method based on alternate optimization. The simulation results show that employing improper signaling improves the achievable rate at medium and high IRI.
\end{abstract}

%\begin{IEEEkeywords}
%Improper Gaussian signaling, asymmetric signaling, alternate relaying, two-path relaying, decode-and-forward, inter-rely-interference.
%\end{IEEEkeywords}

%%%%%%%%%%%%%%%%%%%%%%%%%%%%%%%%%%%%%%%%%%%%%%%%%%%%%%%%%%%%
%%%%%%%%%%%%%%%%%%%%%%%%%%%          Introduction        %%%%%%%%%%%%%%%%%%%%%%
%%%%%%%%%%%%%%%%%%%%%%%%%%%%%%%%%%%%%%%%%%%%%%%%%%%%%%%%%%%%
\section{Introduction}
Next generation wireless communication adopts technologies that extends the network coverage and improve the data rate. One of the candidate technologies is full-duplex relaying that targets to double the spectral efficiency. On the other hand, cooperative communication is an interesting technology to improve the data rate and extend the communication range. Full-duplex relaying is employed to extend the network coverage while improving the link quality. Despite of the promising performance that full-duplex can achieve, replacing all half-duplex nodes by full-duplex ones is not possible to be done immediately. During the roll-out phase, half-duplex nodes are used to support full-duplex services. Two path relaying, which is also known as, \emph{alternate relaying}, is a distributed realization of full-duplex relaying. Full-duplex relaying suffers from self-interference, whereas the two-path relaying suffers from inter-relay interference (IRI). Therefore, different interference mitigation techniques need to be adopted to relief the effect of the interference \cite{ju2009catching}.

Improper signaling is used to mitigate the interference impact on communication systems. It is an asymmetric Gaussian signaling scheme that assumes unequal power of the real and imaginary components and/or dependent real and imaginary  components. It is used in underlay cognitive radio \cite{Lameiro_SS_WCL15,Gaafar2015Spectrum,lameiro2016maximally,amin2016underlay,gaafar2017underlay}, overlay cognitive radio  \cite{amin2017overlay}, full-duplex relaying \cite{gaafar2016improper}, Z-interference channel \cite{lagen2014improper,kurniawan2015improper} and asymmetric hardware distortions \cite{sidrah2017asymmetric}. Recently, we considered the two-path relaying network and  showed that improper signaling can be advantageous over proper signaling to mitigate the IRI  \cite{gaafar2016letter,gaafar2017WSA}.  Specifically, in \cite{gaafar2016letter}, improper signaling is adopted in two-path relaying system, where only the same circularity coefficient, a measure of the degree of impropriety of the signal, for both relays is optimized to mitigate the interference while the relays use  their maximum power. On the other hand, in \cite{gaafar2017WSA}, we considered the same problem but with different circularity coefficients at the relays. Moreover, we considered asymmetric time allocation for the two transmission phases while the relays use their maximum power. 

In this paper, we take the problem in  \cite{gaafar2016letter}  further and optimize both the relay power and  circularity coefficient, which measures the degree of impropriety of the transmit signal, to maximize the  end-to-end achievable rate of the two-path relaying system. First, we consider proper signaling and  introduce optimal relays power allocation for the system. In the case of using improper signaling, we allocate the relays power with a fixed circularity coefficient. Moreover, we tune the circularity coefficient while fixing the transmit power. Then, we jointly optimize the relays power and circularity coefficient via a coordinate descend based method by iterating between the optimal solutions of the individual problems till a convergence obtained. Finally, we investigate through numerical results the merits that can be reaped if the relays use improper signals using different strategies.

%%%%%%%%%%%%%%%%%%%%%%%%%%%%%%%%%%%%%%%%%%%%%%%%%%%%%%%%%%%%
%%%%%%%%%%%%%%%%%%%%%%%%%%%        System Model        %%%%%%%%%%%%%%%%%%%%%%
%%%%%%%%%%%%%%%%%%%%%%%%%%%%%%%%%%%%%%%%%%%%%%%%%%%%%%%%%%%%
%
\section{System Model}\label{sec:sys_mod}
We consider here an alternate two-path relaying network consisting of one source node, $S$, two half-duplex relay nodes, $R_1$ and $R_2$, and one destination node, $D$, as shown in Fig. \ref{fig1}. We adopt decode-and-forward protocol at both relays. Moreover,  the relays transmit and receive in turn, i.e., in one time slot one relay receives and the other relay transmits, and in the next time slot the other way around.  Let $h_i$ and $g_i$, $i \in \{1, 2\}$, denote the channel between $S$ and $R_i$ and the channel between $R_i$ and $D$, respectively. We assume channel reciprocity for the inter-relay channel which is denoted by $f$. Moreover, let us assume that the source transmit power is $p_{\rm{s}}$, the relay transmit power is $p_{\rm{r}}$, and the noise variance at each receiving node is $\sigma_{\rm{n}}^2$. The transmit powers are limited to a power budget of $p_{\rm{max}}$.  First, we give the following definitions of improper random variables (RV).
\begin{definition}\cite{Neeser1993proper} \label{def1}
The complementary (pseudo-) variance of a zero mean complex random variable $x$ is defined as $\tilde{\sigma}_x^2=E\{x^2\}$,   where $\mathbb{E}\{.\}$ denotes the expectation operator. If $\tilde{\sigma}_x^2=0$, then $x$ is called proper signal, otherwise it is called improper.
\end{definition}
\begin{definition} \cite{Lameiro_SS_WCL15}\label{def2}
The circularity coefficient of the signal $x$ is a measure of its impropriety degree and is defined as $\mathcal{C}_x ={|\tilde{\sigma}_x^2|}/{\sigma_x^2}$, where ${\sigma}_x^2=E\{|x|^2\}$ is the conventional variance and $|.|$ is the absolute value operation. The circularity coefficient satisfies $0\leq \mathcal{C}_x \leq 1$.  In particular, $\mathcal{C}_x = 0$ and $\mathcal{C}_x = 1$ correspond to proper and maximally improper signals, respectively .
\end{definition}
We assume that no channel state information is available at $S$ which necessitates the use of proper signaling at $S$ and also makes dirty paper coding of no benefit to fully cancel the IRI. Also, we assume that no direct link is available between $S$ and $D$. For simplicity and tractability, we consider a yet illustrative scenario by assuming equal power and same circularity coefficient for the relays which may  not be optimal. However, as it will be shown in the simulation results, though these sub-optimal assumptions, improper signals show a significantly better performance than proper signaling. Furthermore, we expect even better performance if we increase the degrees of freedom by letting different power and circularity coefficient at the relays. Also, we assume the receivers use the simple practical decoding techniques by treating the interference as a Gaussian noise.
%Moreover, we assume that there is no direct link between $S$ and $D$. Also, we assume that no CSI is available at $S$ and hence, we use proper signals at its transmitter. For simplicity and tractability, we consider a basic yet illustrative scenario by assuming the use of improper signals at $R_i$ with the same circularity coefficient $\mathcal{C}_x$ for each relay. Furthermore, we assume that both $S$ and $R_i$ transmit with a fixed equal power $p_{\rm{r}}$.

%%
%Recall that PGS is capacity-achieving in the case of interference-free systems. On the other hand, IGS is known to improve the achievable rate in interference-limited systems. Hence, in the following, we assume that the source transmits proper signals as it does not cause interference in the network. However, since the signal from one relay causes interference at the other relay in the AR scheme, we propose that the relays send improper signals so that the network can achieve a higher rate.

%achieve a h    the  and  that To reduce IRI, we propose to use improper signaling at the relays. The idea is that the relay receives proper signals but transmit improper signals since it was shown that improper signals cause less interference.
%
%Although proper signaling is a special case of improper signaling, for the sake of clarity, in the following we provide the system model for proper and improper signaling separately.
%
\begin{figure*}[!t]
\centering
\includegraphics[width=5.5in]{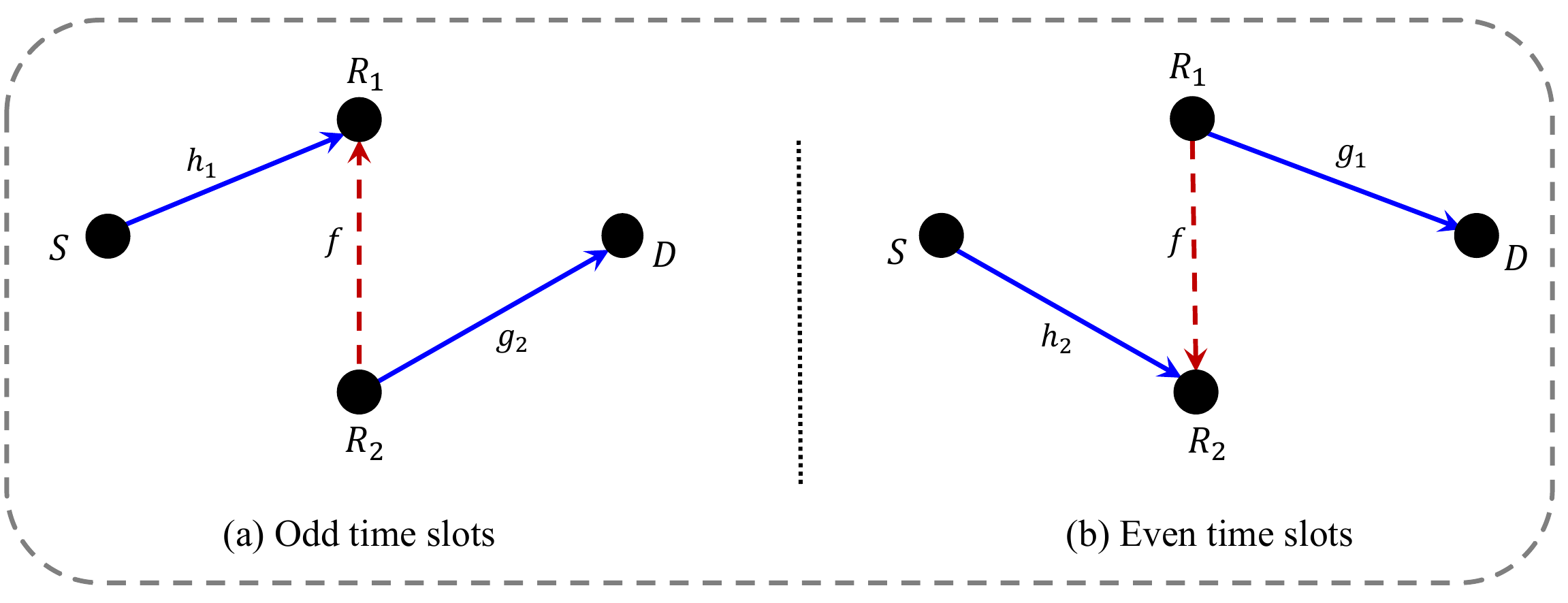}
\caption{Two-Path Relay Channel. The blue solid lines represent the signal links and the red dashed lines represent the IRI.}
\label{fig1}
\end{figure*}

%
%
%%===============================================
%\subsection{Proper Signaling}
%%===============================================
%During time slot $k$, the signal received at $R_i$ with $i=2-\mod(k,2)$ is given by
%\begin{equation}\label{eq1}
%    y_i[k] = \sqrt{p_s}h_i[k]x[k]+\sqrt{p_r}f[k]x[k-1]+n_i[k],
%\end{equation}
%where $x[k]$ is the transmit proper signal in time slot $k$ with $E[|x[k]|^2]=1$ and $n_i[k]$ is a zero-mean additive white Gaussian noise at $R_i$ with variance $\sigma_n^2$. The received signal at the destination from $R_i$ in time slot $k+1$ is given by
%\begin{equation}\label{eq2}
%    r[k+1] = \sqrt{p_r}g_i[k+1]x[k]+n[k+1],
%\end{equation}
%where $n[k+1]$ is a zero-mean additive white Gaussian noise at the destination with variance $\sigma_n^2$.
%%===============================================
%\subsection{Improper Signaling}
%%===============================================
During time slot $k$, the signal received at $R_i$ with $i=2-\mod(k,2)$ is given by\footnote{For the rest of the paper, we let $i,j \in \{1,2\}$, $i \neq j$}
\begin{equation}\label{eq1}
    y_i[k] = \sqrt{p_s} h_i[k] s[k]+\sqrt{p_r}f[k] x_{j}[k]+n_i[k],
\end{equation}
where $s[k]$ is the transmit proper signal by $S$  in time slot $k$ and $n_i[k]$ is the additive noise at $R_i$ with variance $\sigma_{\rm{n}}^2$. $x_j[k]$ is the improper signal, with circularity coefficient $\mathcal{C}_x$, transmitted by $R_j$ with $j=1+\mod(k,2)$. The received signal at $D$ from $R_i$ in time slot $k+1$ is given by
\begin{equation}\label{eq2}
    y_D[k+1] = \sqrt{p_r}g_i[k+1] x_{i}[k+1]+n[k+1],
\end{equation}
where $n[k+1]$ is the additive noise at $D$.
In the following, we assume the channels to be quasi-static block flat fading channels and therefore we drop the time index $k$ for notational convenience. The additive noise at the receivers is modeled as a white, zero-mean, circularly symmetric, complex Gaussian with variance $\sigma_{{n}}^2$.

The alternating two-path relaying system mimics a full-duplex system by transferring the data through two Z-interference channels,  where two transmitters ($S$ and $R_i$) are sending messages each intended for one of the two receivers ($R_j$  and $D$) as shown in Fig. \ref{fig1}. Hence, as a result of using improper signals at $R_j$ and proper signals at $S$  while treating the interference as a Gaussian noise, the achievable rate of the first hop of the $i$th path ($S-R_i$) can be expressed after some simplification steps as  \cite{zeng2013transmit}
\begin{align}\label{rate_first_hop}
{\mathcal{R}_{i,1}}\left( {{p_{\rm{r}}},{{\cal C}_x}} \right) = {\log _2}\left( {1 + \frac{{{p_{\rm{s}}}{{\left| {{h_i}} \right|}^2}}}{{{p_{\rm{r}}}{{\left| f \right|}^2} + \sigma _n^2}}} \right)\hspace{-3pt} + \hspace{-3pt} \frac{1}{2}{\log _2}\left( {\frac{{1 - {\cal C}_{{{{y}}_i}}^2}}{{1 - {\cal C}_{{\cal{I}}_i}^2}}} \right),
\end{align}
where ${\cal C}_{{{{y}}_i}}$ and ${\cal C}_{{\cal I}_i}$ are the circularity coefficients of the received and interference-plus-noise signals at $R_i$, respectively, which can be calculated as
\begin{align}
{{\cal C}_{{{{y}}_i}}} = \frac{{{p_{\rm{r}}}{{\left| f \right|}^2}{{\cal C}_x}}}{{{p_{\rm{s}}}{{\left| {{h_i}} \right|}^2} + {p_{\rm{r}}}{{\left| f \right|}^2} + \sigma _n^2}},\;\;{{\cal C}_{{\cal I}_i}} = \frac{{{p_{\rm{r}}}{{\left| f \right|}^2}{{\cal C}_x}}}{{{p_{\rm{r}}}{{\left| f \right|}^2} + \sigma _n^2}}.
\end{align}
Hence, \eqref{rate_first_hop} can be simplified to
\begin{align}\label{rate_first_hop_sim}
{\mathcal{R}_{i,1}}\left( {{p_{\rm{r}}},{{\cal C}_x}} \right) &= \frac{1}{2} \times \nonumber \\
 & {\log _2}\left( 1 + \frac{{2{p_{\rm{s}}}{{\left| {{h_i}} \right|}^2}\left( {{p_{\rm{r}}}{{\left| f \right|}^2} + \sigma _n^2} \right) + p_{\rm{s}}^2{{\left| {{h_i}} \right|}^4}}}{{\left( {1 - {\cal C}_x^2} \right)p_{\rm{r}}^2{{\left| f \right|}^4} + 2{p_{\rm{r}}}{{\left| f \right|}^2}\sigma _n^2 + \sigma _n^4}} \right).
\end{align} 
Similarly, the achievable rate of the second hop of the $i$th path can be obtained from \eqref{eq2} as
\begin{equation}\label{rate_second_hop}
{\mathcal{R}_{i,2}}\left( {{{p_{\rm{r}},\cal C}_x}} \right) = {\log _2}\left( {1 + \frac{{{p_{\rm{r}}}{{\left| {{g_i}} \right|}^2}}}{{\sigma _n^2}}} \right) + \frac{1}{2}{\log _2}\left( {\frac{{1 - {\cal C}_{{y_{{D}}}}^2}}{{1 - {\cal C}_{{{\cal I}_{{D}}}}^2}}} \right),
\end{equation}
where ${\cal C}_{{{{y}}_D}}$ and ${\cal C}_{{\cal I}_D}$ are the circularity coefficients of the received and interference-plus-noise signals at $D$, respectively, which can be computed as
\begin{equation}
{{\cal C}_{{y_{{D}}}}} = \frac{{{p_{\rm{r}}}{{\left| {{g_i}} \right|}^2}{{\cal C}_x}}}{{{p_{\rm{r}}}{{\left| {{g_i}} \right|}^2} + \sigma _n^2}},\quad {{\cal C}_{{{\cal I}_{{D}}}}} = 0.
\end{equation}
Then, \eqref{rate_second_hop} reduces to
\begin{equation}\label{rate_second_hop_sim}
{\mathcal{R}_{i,2}}\left( {{p_{\rm{r}}},{{\cal C}_x}} \right) = \frac{1}{2}{\log _2}\left( {1 + \frac{{2{p_{\rm{r}}}{{\left| {{g_i}} \right|}^2}}}{{\sigma _n^2}} + \frac{{p_{\rm{r}}^2{{\left| {{g_i}} \right|}^4}\left( {1 - {\cal C}_x^2} \right)}}{{\sigma _n^4}}} \right).
\end{equation}
Hence, the  end-to-end achievable rate of the $i$th path can be calculated from
\begin{equation}
{\mathcal{R}_{{_i}}}\left( {{p_{\rm{r}}},{{\cal C}_x}} \right) = \min \Big\{ {{\mathcal{R}_{i,1}}\left( {{p_{\rm{r}}},{{\cal C}_x}} \right),{\mathcal{R}_{i,2}}\left( {{p_{\rm{r}}},{{\cal C}_x}} \right)} \Big\}.
\end{equation} 
Accordingly, the overall end-to-end achievable rate of the two-path relaying system, for sufficiently large number of time slots\footnote{One slot is missed at the start of the transmission without delivering information from $S$ to $D$.}, is expressed as the arithmetic mean of ${{\mathcal{R}_{{_i}}}\left( {{p_{\rm{r}}},{{\cal C}_x}} \right)}$ 
\begin{equation}\label{R_tot}
{\mathcal{R}_{{\rm{T}}}}\left( {{p_{\rm{r}}},{{\cal C}_x}} \right) = \frac{1}{2}\sum\limits_{i = 1}^2 {{\mathcal{R}_{{_i}}}\left( {{p_{\rm{r}}},{{\cal C}_x}} \right)}.
\end{equation} 

\begin{remark}
One can notice that if $\mathcal{C}_x=0$ in \eqref{R_tot}, we obtain the conventional expression for the total achievable rate of the two-path relaying system under the use of proper signals as
\begin{align}\label{R_tot_p}
&{\mathcal{R}_{{\rm{T}}}}\left( {{p_{\rm{r}}},0} \right) = \frac{1}{2}\sum\limits_{i = 1}^2 {{\mathcal{R}_{{_i}}}\left( {{p_{\rm{r}}},0} \right)}= \frac{1}{2} \times\nonumber \\
&\sum\limits_{i = 1}^2 \min \Bigg\{ {{\log }_2}\left( {1 + \frac{{{p_{\rm{s}}}{{\left| {{h_i}} \right|}^2}}}{{\sigma _n^2 + {p_{\rm{r}}}{{\left| f \right|}^2}}}} \right),  {{\log }_2}\left( {1 + \frac{{{p_{\rm{r}}}{{\left| {{g_i}} \right|}^2}}}{{\sigma _n^2}}} \right) \Bigg\}. 
\end{align}  
\end{remark}

%%%%%%%%%%%%%%%%%%%%%%%%%%%%%%%%%%%%%%%%%%%%%%%
%%%%%%%%%%%%%%%         Performance Analysis       %%%%%%%%%%%%%%%%
%%%%%%%%%%%%%%%%%%%%%%%%%%%%%%%%%%%%%%%%%%%%%%%
\section{Improper Gaussian Signaling Design for Two-Path Relaying Systems}\label{sec:analysis}
In this section, we aim at optimizing the relays signal parameters represented in the relay's transmit power $p_{\rm{r}}$ and the circularity coefficient $\mathcal{C}_x$  in order to maximize the instantaneous end-to-end achievable rate of the system. First, the intuition behind the benefit of using improper signals at the relays is that it provides an additional degree of freedom that can be optimized in order to alleviate the effect of the IRI on the relays or, in the worst case, kept at the same performance as proper signaling, i.e., $\mathcal{C}_x=0$. Moreover, improper signaling has the ability to control the interference signal dimension, and it is one form of interference alignment \cite{kurniawan2015improper,cadambe2008interference}. Furthermore, when using proper signals, $R_i$ can improve the rate of the second hop of the $i$th by boosting its transmit power. However, this will deteriorate the rate of the first hop of the $j$th path and here improper signaling attains its benefit. By increasing the asymmetry of the relay's transmit signal, by boosting the circularity coefficient, the relay can increase its power and has a less adverse effect on the other one.

Now, in order to reap the benefits of improper signaling, we design the power and circularity coefficient of the relays.  For this purpose, we formulate the following optimization problem
\begin{align}
{\bf{P1}}:&\mathop {\max }\limits_{{p_{\rm{r}}},\mathcal{C}_x}\qquad  {\mathcal{R}_{{\rm{T}}}}\left( {{p_{\rm{r}}},\mathcal{C}_x} \right) \nonumber \\
&\;{\rm{s}}{\rm{.t}}{\rm{.}}\quad \quad \;\;0 < {p_{\rm{r}}} \le {p_{\max }},\nonumber\\
& \qquad \qquad \; 0\leq \mathcal{C}_x \leq 1.
\end{align}
Solving ${\bf{P1}}$ optimally is difficult as it is a non-convex optimization problem. Here, we propose a coordinate-descent (CD) based method in which we consider two problems, optimizing the relays transmit power for a fixed circularity coefficient and optimizing the circularity coefficient for a fixed transmit power. Finally, we perform alternate optimization of the optimal solutions of the two problems till we get convergence.
\begin{remark} \cite{Bertsekas1999nonlinear}\label{convergence_CD}
The CD method is popular for its efficiency, simplicity and scalability. Moreover, it is guaranteed to converge to a local solution if the global optimal solution is attained for each of the sub-problems. However, it does not necessarily converge to the global optimal solution as the objective function is non-convex. 
\end{remark} 
Following Remark \ref{convergence_CD}, we will show the optimal solutions of the two sub-problems. First, for notational convenience, we give the following definitions.
\begin{definition}\label{permutation_def}
Let $\pi$ denote the permutation of $\{1,2\}$ that sets the points ${z_i} \in \mathbb{R}_{++}$ in an increasing order such that $z_{\pi_1} \leq z_{\pi_2}$. Also, let $\mathcal{F}_{i,j}\left( {{p_{\rm{r}}},{{\cal C}_x}} \right) ={\mathcal{R}_{i,1}}\left( {{p_{\rm{r}}},{{\cal C}_x}} \right) +{\mathcal{R}_{j,2}}\left( {{p_{\rm{r}}},{{\cal C}_x}} \right)$ and \\[0.25cm]
$k_i\left( {{p_{\rm{r}}},{{\cal C}_x}} \right)=\mathop {\arg \min }\limits_{a \in \{1,2\}} {\mathcal{R}_{i,a}}\left( {{p_{\rm{r}}},{{\cal C}_x}} \right)$. 
\end{definition}
\vspace{2pt}

\noindent \textit{Sub-problem 1) }\textit{Relays Transmit Power Optimization Problem}

In this part, we optimize the relays transmit power for a fixed circularity coefficient $\mathcal{C}_x^o$. The corresponding optimization problem is given by
\begin{align}
{\bf{P2}}\left(\mathcal{C}_x^o\right):&\mathop {\max }\limits_{{p_{\rm{r}}}}\qquad  {\mathcal{R}_{{\rm{T}}}}\left( {{p_{\rm{r}}},\mathcal{C}_x^o} \right)\nonumber \\
&\;{\rm{s}}{\rm{.t}}{\rm{.}}\quad \quad \;\; 0 < {p_{\rm{r}}} \le {p_{\max }}.
\end{align}               
It can be verified that ${\bf{P2}}$ is a non-convex optimization problem which makes it hard, in general, to find its optimal solution. Also, due to the coupling between the achievable rates of the two paths in terms of $p_{\rm{r}}$, maximizing the rates of each individual path with respect to $p_{\rm{r}}$ and taking the arithmetic mean is not optimal. However, thanks to some special monotonicity properties of the objective function, we show that the optimal solution of ${\bf{P2}}$ lies  either at the intersection between ${\mathcal{R}_{i,1}}\left( {{p_{\rm{r}}},\mathcal{C}_x^o}\right)$  and ${\mathcal{R}_{i,2}}\left( {{p_{\rm{r}}},\mathcal{C}_x^o}\right)$, if exists or one of the stationary points of the $\mathcal{F}_{i,j}\left( {{p_{\rm{r}}},{{\cal C}_x^o}} \right)$ with respect to $p_{\rm{r}}$, if exists or the power budget  $p_{\rm{max}}$. Next, we will compute the intersection and stationary points.
\begin{proposition}\label{prop1}
There exists at most one intersection point, ${p_i}$, between ${\mathcal{R}_{i,1}}\left( {{p_{\rm{r}}},\mathcal{C}_x^o}\right)$  and ${\mathcal{R}_{i,2}}\left( {{p_{\rm{r}}},\mathcal{C}_x^o}\right)$ over the feasible interval $0 < p_{\rm{r}} \leq p_{\rm{max}}$. Moreover, this intersection point can be obtained by solving the quartic equation\footnote{The quartic equation can be solved by Ferrari's method \cite{gerolamo1993ars}. However, since the roots derived from this quartic equation are extremely complex and lengthy, we omit them due to the space limitations.}in \eqref{eq_quartic}.
\begin{figure*}
\begin{align}\label{eq_quartic}
\frac{{{{\left| {{g_i}} \right|}^4}{{\left| f \right|}^4}}{\left( {1 - {\cal C}_x^{o^2}} \right)^2}}{{\sigma _n^4}}p_i^4 + \frac{{{{2\left| {{g_i}} \right|}^2}{{\left| f \right|}^2}\left( {{{\left| {{g_i}} \right|}^2} + {{\left| f \right|}^2}} \right)}\left( {1 - {\cal C}_x^{o^2}} \right)}{{\sigma _n^2}}&p_i^3 + {\left| {{g_i}} \right|^2}\left( {4{{\left| f \right|}^2} + {{\left| {{g_i}} \right|}^2}\left( {1 - {\cal C}_x^{o^2}} \right)}  \right)p_i^2 \nonumber \\
& \hspace{-5pt}+ 2\left( {\sigma _n^2{{\left| {{g_i}} \right|}^2} - {p_{\rm{s}}}{{\left| {{h_i}} \right|}^2}{{\left| f \right|}^2}} \right){p_i} -  \left( {p_{\rm{s}}^2{{\left| {{h_i}} \right|}^4} + 2{p_{\rm{s}}}{{\left| {{h_i}} \right|}^2}\sigma _n^2} \right) = 0.
\end{align}
\end{figure*}  
 
\end{proposition}
\begin{proof}
By equating ${\mathcal{R}_{i,1}}\left( {{p_{\rm{r}}},\mathcal{C}_x^o}\right)$  and ${\mathcal{R}_{i,2}}\left( {{p_{\rm{r}}},\mathcal{C}_x^o}\right)$, we obtain \eqref{eq_quartic}. Then, by arranging the coefficients of the quartic equation in a descending order, the signs of theses coefficients, according to the sign of the linear term is either $\{+,+,+,+,-\}$ or $\{+,+,+,-,-\}$. In both cases,   there is only one change of signs. For our real quartic polynomial, this determines the number of positive roots to be exactly one root over $\mathbb{R}_{++}$ by using Descartes rule of signs \cite{prasolov2009polynomials}. Hence, there exists at most one intersection point over the feasible interval.
\end{proof}
\begin{remark}
For the case of using proper signals at the relays i.e., $\mathcal{C}_x^o=0$, the quartic equation reduces to the following quadratic equation
\begin{equation}
p_{_i}^2 + \frac{{\sigma _n^2}}{{{{\left| f \right|}^2}}}{p_i} - \frac{{{p_{\rm{s}}}{{\left| {{h_i}} \right|}^2}\sigma _n^2}}{{{{\left| {{g_i}} \right|}^2}{{\left| f \right|}^2}}} = 0,
\end{equation}
which can be solved to obtain the intersection point as
\begin{equation}
{p_i} = \frac{1}{{2\left| f \right|}}\left( {\sqrt {\frac{{\sigma _n^4}}{{{{\left| f \right|}^2}}} + \frac{{4{p_{\rm{s}}}{{\left| {{h_i}} \right|}^2\sigma _n^2}}}{{{{\left| {{g_i}} \right|}^2}}}}  - \frac{{\sigma _n^2}}{{\left| f \right|}}} \right).
\end{equation}
\end{remark}
Now, We can divide $\mathbb{R}_{++}$ into three intervals where $\{0, p_{\pi_1}, p_{\pi_2}, \infty\}$ are the boundaries for these intervals. From \eqref{R_tot}, ${\mathcal{R}_{{\rm{T}}}}\left( {{p_{\rm{r}}},\mathcal{C}_x^o} \right)$ can be reformulated in each interval as 
\begin{align}\label{R_tot_IGS_pr}
&{\mathcal{R}_{{\rm{T}}}}\left( {{p_{\rm{r}}},\mathcal{C}_x^o} \right) = \frac{1}{2} \times\nonumber \\
& \left\{ {\begin{array}{*{20}{c}}
{\sum\limits_{i = 1}^2 {\mathcal{R}_{i,2}}\left( {{p_{\rm{r}}},\mathcal{C}_x^o}\right), }&{\rm{if}}&{0 < {p_{\rm{r}}} \le p_{\pi_1}}\\
{\mathcal{R}_{\pi_1,1}}\left( {{p_{\rm{r}}},\mathcal{C}_x^o}\right)+{\mathcal{R}_{\pi_2,2}}\left( {{p_{\rm{r}}},\mathcal{C}_x^o}\right),&{\rm{if}}&{p_{\pi_1}  < {p_{\rm{r}}} \le p_{\pi_2} }\\
{\sum\limits_{i = 1}^2 {\mathcal{R}_{i,1}}\left( {{p_{\rm{r}}},\mathcal{C}_x^o}\right), }&{\rm{if}}&{p_{\pi_2}  < {p_{\rm{r}}} < \infty }
\end{array}} \right..
\end{align}

There are at maximum five stationary points $p^{(n)}_{{\rm{st}}_i} \in \mathbb{C}$, $n \in \{1,2,3,4,5\}$ of  $\mathcal{F}_{i,j}\left( {{p_{\rm{r}}},{{\cal C}_x}} \right)$ which can be calculated by finding the roots of  the derivative of $\mathcal{F}_{i,j}\left( {{p_{\rm{r}}},{{\cal C}_x}} \right)$ with respect to $p_{\rm{r}}$ over the  interval $0 < p_{\rm{r}} \leq p_{\rm{max}}$. The resulting equation is a quintic equation\footnote{The feasible roots of the quintic equation can be obtained numerically. }, which is very lengthy and we omit it due to space limitation.

%Before introducing the optimal solution of ${\bf{P2}}$, let us define the set of powers ${\hat {\cal P}}=\{p_{\pi_1},{p_{{\rm{st},k}}},p_{\pi_2},p_{\rm{max}}\}$, where ${p_{{\rm{st},k}}}$, $k \in \{1,2,3,4,5\}$ are the \textit{possible} stationary points, if exist,  of ${R_{{\rm{T}}}}\left( {{p_{\rm{r}}},\mathcal{C}_x^o} \right)$ in the middle interval, i.e., ${p_{\pi_1}  < {p_{\rm{r}}} \le p_{\pi_2} }$. We compute these stationary points by   equating the derivative, with respect to $p_{\rm{r}}$, of the function in the middle interval with zero. The resulting equation is a quintic equation\footnote{The feasible roots of the quintic equation, in the corresponding interval, can be obtained numerically. }, which is very lengthy and we omit it due to space limitation.
\begin{remark}
When using proper signals at the relays, the quintic equation reduces to the quadratic equation
\begin{equation}
\frac{{{{\left| {{g_j}} \right|}^2}}}{{\sigma _n^2}}{\left( {\sigma _n^2 + {p_{{{\mathrm{st}}}_i}}{{\left| f \right|}^2}} \right)^2} = {p_{\rm{s}}}{\left| {{h_i}} \right|^2}\left( {{{\left| f \right|}^2} - {{\left| {{g_j}} \right|}^2}} \right),
\end{equation}
 which can be solved to get only one possible stationary point as
\begin{equation}
{p_{{{\mathrm{st}}}_i}} = \sqrt {\frac{{\sigma _n^2{p_{\rm{s}}}{{\left| {{h_i}} \right|}^2}}}{{{{\left| {{g_j}} \right|}^2}{{\left| f \right|}^4}}}\left( {{{\left| f \right|}^2} - {{\left| {{g_j}} \right|}^2}} \right)}  - \frac{{\sigma _n^2}}{{{{\left| f \right|}^2}}},
\end{equation}
in which it can be easily shown that ${p_{{{\mathrm{st}}}_i}} \in \mathbb{R}_{++}$ if and only if
\begin{equation}
{\left| f \right|^2} - {\left| {{g_j}} \right|^2} > \frac{{\sigma _n^2{{\left| {{g_j}} \right|}^2}}}{{{p_{\rm{s}}}{{\left| {{h_i}} \right|}^2}}}.
\end{equation}
\end{remark}
Before introducing the optimal solution of ${\bf{P2}}$, let us give the following definition
\begin{definition}\label{def_feasible}
Let the set of feasible transmit powers $\mathcal{P}_{\rm{int}}=\left\{ p_i \mid {0  < p_i\le p_{\rm{max}} } \right\}$. Also, the set of feasible stationary points $\mathcal{P}_{\rm{st}}=\left\{ p^{(n)}_{{\rm{st}}_i} ,\mid {p_{\pi_1}  < p^{(n)}_{{\rm{st}}_i}\le p_{\pi_2} }\; \&\; p^{(n)}_{{\rm{st}}_i} \leq p_{\rm{max}}\right\}$. 
\end{definition}

From Definition \ref{def_feasible}, $\mathcal{P}_{\rm{int}}$ and $\mathcal{P}_{\rm{st}}$ can be empty sets. Based on the aforementioned analysis, the optimal solution of  ${\bf{P2}}$ can be found from the following theorem.
\begin{theorem}\label{theorem_IGS_pr}
In a two-path relaying system, where the two relays transmit improper signals and by treating interference as a Gaussian noise, the optimal power allocation, at a fixed circularity coefficient, that maximizes the total achievable rate constrained by a power budget $p_{\rm{max}}$ can be obtained as
\begin{equation}
\begin{array}{l}
p_{\rm{r}}^* = \mathop {\arg \max }\limits_{p \in \mathcal{P}_{{\rm{T}}}} {\mathcal{R}_{{\rm{T}}}}\left( {p,\mathcal{C}_x^o} \right),
\end{array}
\end{equation}
where $\mathcal{P}_{{\rm{T}}}=\mathcal{P}_{\rm{int}} \cup \mathcal{P}_{\rm{st}} \cup p_{\rm{max}}$.
\end{theorem}

\begin{proof}
 From the definition of the total rate function in \eqref{R_tot_IGS_pr}, it can be readily verified that the function in the first interval, i.e., ${0 < {p_{\rm{r}}} \le p_{\pi_1} }$, is monotonically increasing in $p_{\rm{r}}$, thus the optimal solution of $\bf{P2}$ in this interval is $p_{\pi_1}$.  Moreover, the function in \eqref{R_tot_IGS_pr} in the last interval, i.e., ${p_{\pi_2}  < {p_{\rm{r}}} < \infty }$, is monotonically decreasing in $p_{\rm{r}}$ and hence the optimal solution in this interval is $p_{\pi_2}$. If the maximum of ${R_{{\rm{T}}}}$ is in the middle interval, it must occur at  a stationary point. Finally, we limit these points by the power budget and this concludes the proof.  
\end{proof}
\vspace{2pt}

\noindent \textit{Sub-problem 2) }\textit{ Circularity Coefficient Optimization Problem}
Now, we optimize the impropriety of the relays transmit signal, measured by the circularity coefficient, assuming a fixed transmit power ${p_{\rm{r}}^o}$. To this end, we formulate the following optimization problem.
\begin{align}
{\bf{P3}}\left(p_{\rm{r}}^o\right):&\mathop {\max }\limits_{\mathcal{C}_x}\qquad  {\mathcal{R}_{{\rm{T}}}}\left( {{p_{\rm{r}}^o},\mathcal{C}_x} \right)\nonumber \\
&\; {\rm{s}}{\rm{.t}}{\rm{.}} \quad \quad \;\; 0\leq \mathcal{C}_x \leq 1.
\end{align}
This problem has been addressed in our work \cite{gaafar2016letter} and the optimal solution is given in the following theorem.
\begin{theorem}\cite{gaafar2016letter}
In a two-path relaying system, where the two relays transmit improper signals and by treating interference as a Gaussian noise, the optimal circularity coefficient, at a fixed relay transmit power, that maximizes the total achievable rate can be obtained as

\vspace{5pt}
\noindent {\bf Case 1:} no intersection points 
\scriptsize
\begin{align}\mathcal{C}_x^*=
\begin{cases}
{0},&{\rm{if}}\quad k_1\left( {\mathcal{C}}\right)=k_2\left( {\mathcal{C}}\right)=2,\; 0 \leq \mathcal{C} \leq 1\\
{1},&{\rm{if}}\quad k_1\left( {\mathcal{C}}\right)=k_2\left( {\mathcal{C}}\right)=1,\; 0 \leq \mathcal{C} \leq 1\\
{\mathop {\arg \max }\limits_{\mathcal{C}_x \in {\{{0,\mathcal{C}}_{{\rm{st}}_i},1\} }} \mathcal{F}_{i,j}\left({p_{\rm{r}}^o} ,{\mathcal{C}_x}\right)},&{\rm{if}}   \quad k_1\left( {\mathcal{C}}\right)=i, k_2\left( {\mathcal{C}}\right)=j,0 \leq \mathcal{C} \leq 1
\end{cases}.
\end{align}
\normalsize
\noindent {\bf Case 2:} one intersection point, ${{\cal C}_{{i}}}$  
\scriptsize
\begin{align}\mathcal{C}_x^*=
\begin{cases}
{\mathop {\arg \max }\limits_{\mathcal{C}_x \in {\{{{\cal C}_{{i}}},{\mathcal{C}}_{{\rm{st}}_j},1\}}} \mathcal{F}_{j,i}\left({p_{\rm{r}}^o}, {\mathcal{C}_x}\right)},&{\rm{if}}\quad {k_j\left( {\mathcal{C}}\right)=1,\;0 \leq \mathcal{C} \leq 1}\\
{\mathop {\arg \max }\limits_{\mathcal{C}_x \in {\{0,{\mathcal{C}}_{{\rm{st}}_i},{{\cal C}_{{i}}}\}}} \mathcal{F}_{i,j}\left({p_{\rm{r}}^o}, {\mathcal{C}_x}\right)},&{\rm{if}}\quad {k_j\left( {\mathcal{C}}\right)=2,\;0 \leq \mathcal{C} \leq 1}
\end{cases}.
\end{align} 
\normalsize
\noindent {\bf Case 3:} two intersection points, ${\left(\mathcal{C}_{\pi_1},\mathcal{C}_{\pi_2} \right)}$
\scriptsize  
\begin{align}
\mathcal{C}_x^*=\mathop {\arg \max }\limits_{\mathcal{C}_x \in {\{\mathcal{C}_{\pi_1},{\mathcal{C}}_{\rm{st}_{\pi_2}},\mathcal{C}_{\pi_2} \}}} \mathcal{F}_{\pi_2,\pi_1}\left({p_{\rm{r}}^o}, {\mathcal{C}_x}\right).
\end{align}
\normalsize
where $\mathcal{C}_i$ and ${\mathcal{C}}_{{\rm{st}}_i}$ are the intersection between ${\mathcal{R}_{i,1}}\left( {{p_{\rm{r}}^o},\mathcal{C}_x}\right)$  and ${\mathcal{R}_{i,2}}\left( {{p_{\rm{r}}^o},\mathcal{C}_x}\right)$ and the stationary point for $\mathcal{F}_{i,j}\left({p_{\rm{r}}^o},\mathcal{C}_x\right)$ with respect to $\mathcal{C}_x$, over the feasible interval $0 < \mathcal{C}_x \leq 1$, respectively\footnote{For more about the existance and uniqueness of $\mathcal{C}_i$ and ${\mathcal{C}}_{{\rm{st}}_i}$, please refer to \cite{gaafar2016letter}. }.
\end{theorem} 
\begin{proof}
An extended version of the  proof in \cite{gaafar2016letter} is provided in the appendix.
\end{proof}

\vspace{2pt}

\noindent \textit{Coordinate Descent: }\textit{Joint Optimization Problem }

Here, we aim at optimizing jointly the relays power and circularity coefficient in order to maximize the total rate of the two-path relaying system via CD, in which we implement alternate optimization of $p_{\rm{r}}$ and $\mathcal{C}_x$. In this method, we optimize the transmit power for a fixed circularity coefficient. Then, we use the optimal power in the previous step to optimize for the circularity coefficient and iterate between the optimal solutions till a stopping criterion is satisfied. For this purpose, we develop Algorithm I to obtain the optimization parameters of $\bf{P1}$. 

\floatname{algorithm}{}
\begin{algorithm} \label{alg1}
\renewcommand{\thealgorithm}{}
\newcommand{\tab}[1]{\hspace{.06\textwidth}\rlap{#1}}
\caption{\textbf{ Algorithm I: Joint Alternate Optimization of the power and circularity coefficient based on the CD method. }}
\begin{algorithmic}[1]
\State \textbf{Input} $h_i$, $g_i$, $f$, ${\sigma _n^2}$, $p_{\rm{max}}$, $\epsilon_{\rm{max}}$, $0 < {p_{\rm{r}}^o} \leq p_{\rm{max}} $, \qquad \qquad$0 \leq \mathcal{C}_x^o \leq 1$.\vspace{3pt}
\State \textbf{Initialize} ${p_{\rm{r}}}\leftarrow{p_{\rm{r}}^o}$, $\mathcal{C}_x\leftarrow\mathcal{C}_x^o$ and $\epsilon \leftarrow \infty$ \vspace{3pt}
\While{$\epsilon > \epsilon_{\rm{max}}$ }\vspace{7pt}
\State \textbf{Compute} $\hat p_{\textrm{{r}}} $ from  $\textbf{P2} \left( \mathcal{C}_x\right)$ using Theorem 1\vspace{5pt}
\State \textbf{Compute} $ \hat{ \mathcal{C}}_x $ from  $\textbf{P3} \left(\hat p_{\textrm{{r}}} \right)$ using Theorem 2 \vspace{3pt}
\State \textbf{Set} $\epsilon=\max\left\{\mid\hat{\mathcal{C}}_x-\mathcal{C}_x\mid,\mid \hat p_{\textrm{r}}-p_{\textrm{r}}\mid\right\}$ $\%max\; error$\vspace{5pt}
\State \textbf{Update} ${p_{\rm{r}}}\leftarrow{\hat p_{\rm{r}}}$ \vspace{3pt}
\State \textbf{Update} $\mathcal{C}_x\leftarrow\hat{\mathcal{C}}_x$ \vspace{3pt}
\EndWhile

%\State \textbf{Compute} ${{\cal I}_{{\rm{max,}}{{\rm{p}}_i}}}$ from \eqref{Imax}.
%\If {$\left({\cal I}_{{\rm{max,}}{\rm{p}_i}} > {p_{_j}}\upsilon _{{\rm{p}}_j}\right)$}\vspace{0.1cm}
%\For {$z=1:k+1$}
%\State \textbf{Compute} $m = \mathop {\arg \min }\limits_{l \in \{0,1,2\}}   p_{\mathrm{s}}^{( l  )} \left(  \frac{{\cal C}_ x^{(z-1)} +{\cal C}_ x^{(z)}}{2} \right)$ 
%\If{ $m =0$}\vspace{0.1cm}
%\State $p_{\rm{o}}^{(z)}\leftarrow {p_{{\rm{s,max}}}}$,   $ \quad \mathcal{C}_\mathrm{o}^{(z)}\leftarrow {\cal C}_ x ^{(z-1)}$\vspace{0.1cm}
%\ElsIf{ $\left(\frac{{{\beta _j}{\gamma _{\rm{s}}}{\Psi _m}\left( {1,2} \right)}}{{{{\cal I}_{{{\rm{s}}_j}}}\left( {\sum\nolimits_{i = 1}^2 {{p_i}{{\cal I}_{{{\rm{p}}_i}}}}  + 1} \right) }} >  - 1\right)$, $j\neq m$}\vspace{0.1cm}
%\State $p_{\rm{o}}^{(z)}\leftarrow p_\mathrm{s}^{\left( m \right)}\left({\cal C}_x ^{(z)}\right)$, $\quad \mathcal{C}_{\rm{o}}^{(z)}\leftarrow {\cal C}_x ^{(z)}$
%\Else {}
%\State $p_{\rm{o}}^{(z)}\leftarrow p_{\mathrm{s}}^{\left( m \right)}\left({\cal C}_x ^{(z-1)}\right)$, $\quad\mathcal{C}_{\rm{o}}^{(z)}\leftarrow {\cal C}_x ^{(z-1)}$
%\EndIf 
%\EndFor
%\State \textbf{Output} $\left( {p_{\rm{s}}^*,{\cal C}_x^*} \right) = \mathop {\arg \max }\limits_{p_{\rm{o}}^{(z)},\;{\cal C}_o^{(z)}} {R_{{\rm{s}}}}\left( {p_{\rm{o}}^{(z)},{\cal C}_{\rm{o}}^{(z)}} \right)$
%\Else{}
%\State \textbf{Output} $\left( {p_{\rm{s}}^*,{\cal C}_x^*} \right)=\left(0,0\right)$
%\EndIf
%
\end{algorithmic}
\end{algorithm}

\section{Numerical Results}\label{sec:results}
In this section, we numerically evaluate the average end-to-end rate of the proposed two-path relaying system using improper signaling. Throughout the following simulation scenarios, we compare between proper and improper signaling. 
For proper based system system, we include two scenarios: maximum power allocation (MPA) and optimal power allocation (OPA). On the other hand for improper based system, we include three scenarios: MPA for maximally improper relay signal, i.e., $\mathcal{C}_x=1$, optimized CD based method using an initial point for the power as  $p^o_{\textrm{r}}=p_{\textrm{max}}$ and  two different initial starting points for the circularity coefficient; $\mathcal{C}_x^{0} = 0$  and  $\mathcal{C}_x^{0} = 1$ and the joint optimal allocation of $p_{\mathrm{r}}$ and $\mathcal{C}_x$ using a fine exhaustive grid search (GS) as a benchmark for the alternate optimization. The average channel signal-to-noise ratios (SNRs) are defined as $\gamma_{h_i} = {\sigma^2_{h_i}}/{\sigma^2_n}$, $\gamma_{g_i} = {\sigma^2_{g_i}}/{\sigma^2_n}$ and $\gamma_{f} = {\sigma^2_{f}}/{\sigma^2_n}$ . The results are averaged over $10000$ channel realizations and $\epsilon_{\rm{max}}=0.0001$.

 As for the simulation setup, we assume symmetric relays links with zero-mean complex Gaussian distribution and $\gamma_{h_i} = 10 \; \mathrm{dB}$, $\gamma_{g_i} = 15 \; \mathrm{dB}$, $\gamma_{f}=20 \; \mathrm{dB}$, unless otherwise specified.

Firstly, to explore the impact of improper signaling on two-path relaying systems,  we study the average rate performance  versus $\gamma_{f}$ as can be seen in Fig \ref{fig_R1R2}. It is clear that, proper and improper based  systems suffer from a rate degradation as the interference link increases  which worsen the performance of $\mathrm{S}-\mathrm{R}$ links and thus limits the end-to-end  rate. For the proper based system, we observe that optimizing the relay power reduces the IRI impact on the relays and improve the  rate. As for improper signaling, optimizing $\mathcal{C}_x$ with maximum power can significantly boost the rate at mid and high interference levels. At low interference levels, improper-MPA achieves better performance than proper-MPA, however  it can not compete with proper-OPA as the interference is not dominant in such situation and thus proper signaling becomes preferable. The same observation is observed for other improper based systems when compared with proper-MPA. 

As for CD joint optimization solution, the proper choice of initial points in CD plays an important role in the overall performance compared with the GS solution as can be observed in Fig. \ref{fig_R1R2}. As a result, staring the CD with $\mathcal{C}_x^{0} =1$ can converge to the GS solution while  $\mathcal{C}_x^0 =0$ improves the rate performance but it does not converge  to the optimal performance. This observation can be justified as the solution at high interference levels reduces to maximally improper, i.e., $\mathcal{C}_x^*=1$ as can be seen from the improper-MPA system.

\begin{figure}[!t]
\centering
\includegraphics[width=3.5in]{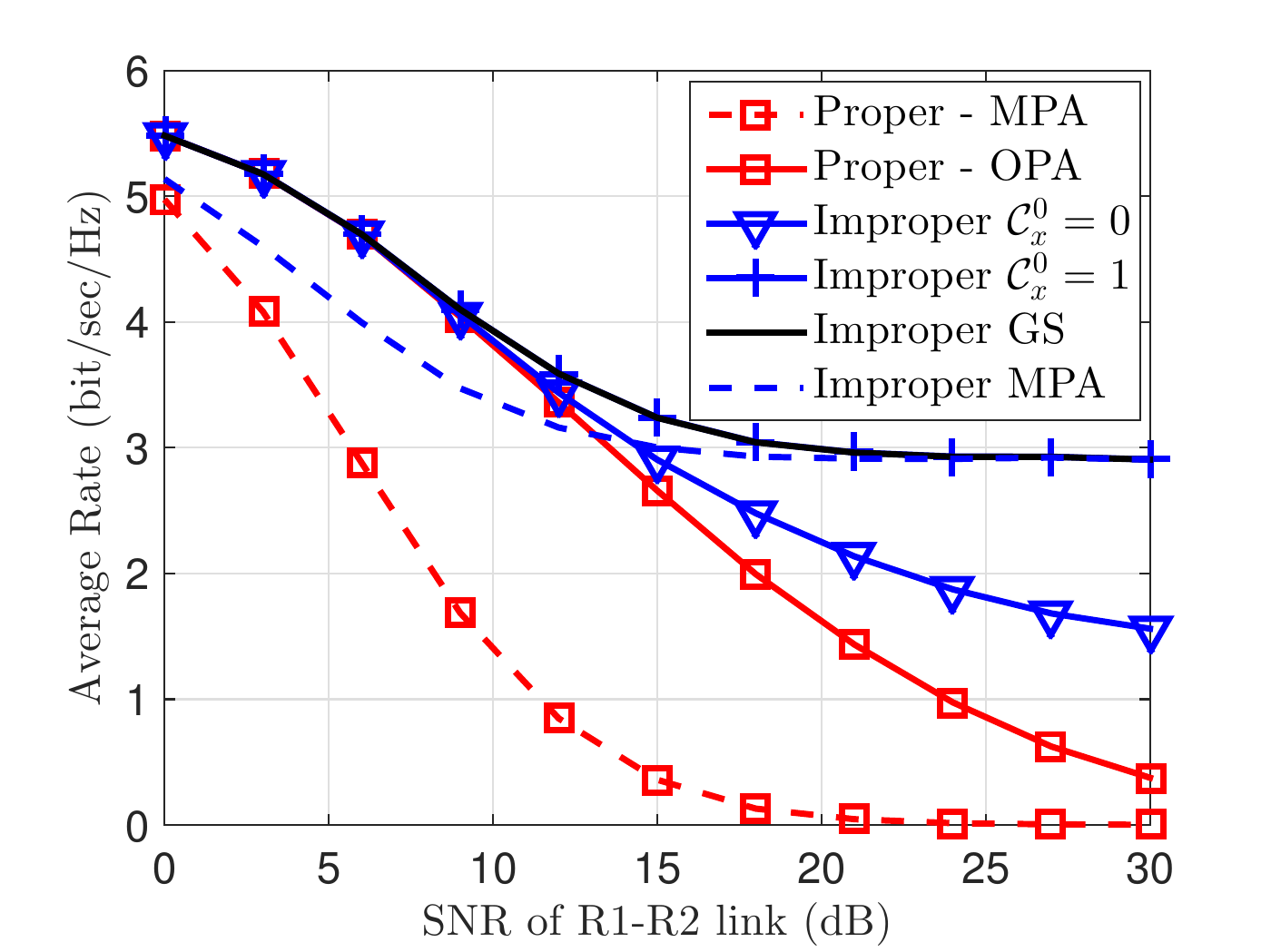}
\caption{The average achievable end-to-end rate for proper and improper signaling with different methods versus $\gamma_{f}$.}
\label{fig_R1R2}
\end{figure}

Secondly, we study the average end-to-end rate performance of the aforementioned system versus $\gamma_{h_i}$ as can be shown in Fig. \ref{fig_SR}. At very low $\gamma_{h_i}$ values, the first hops become a bottleneck and degrade the end-to-end average rate for both proper and improper based systems. As $\gamma_{h_i}$ increases, improper systems use more transmit powers and alleviate the IRI through the increase of the signal impropriety by boosting the circularity coefficient while proper-MPA systems use relatively less power. This improvement gap remains until the value of $\gamma_{h_i}$ becomes relatively large with respect to   $\gamma_{f}$, and hence the proper based system starts to enhance its performance by increasing its transmit power. At high $\gamma_{h_i}$, both systems tend to utilize the power budget and the improper solution reduces to  proper. From this investigation, we can state that improper signaling is preferred when the first hops become a bottleneck. As expected from the previous simulation scenario at $\gamma_{f} = 20 \; \mathrm{dB}$, improper-MPA achieves a close performance to the improper-GS. 

\begin{figure}[!t]
\centering
\includegraphics[width=3.5in]{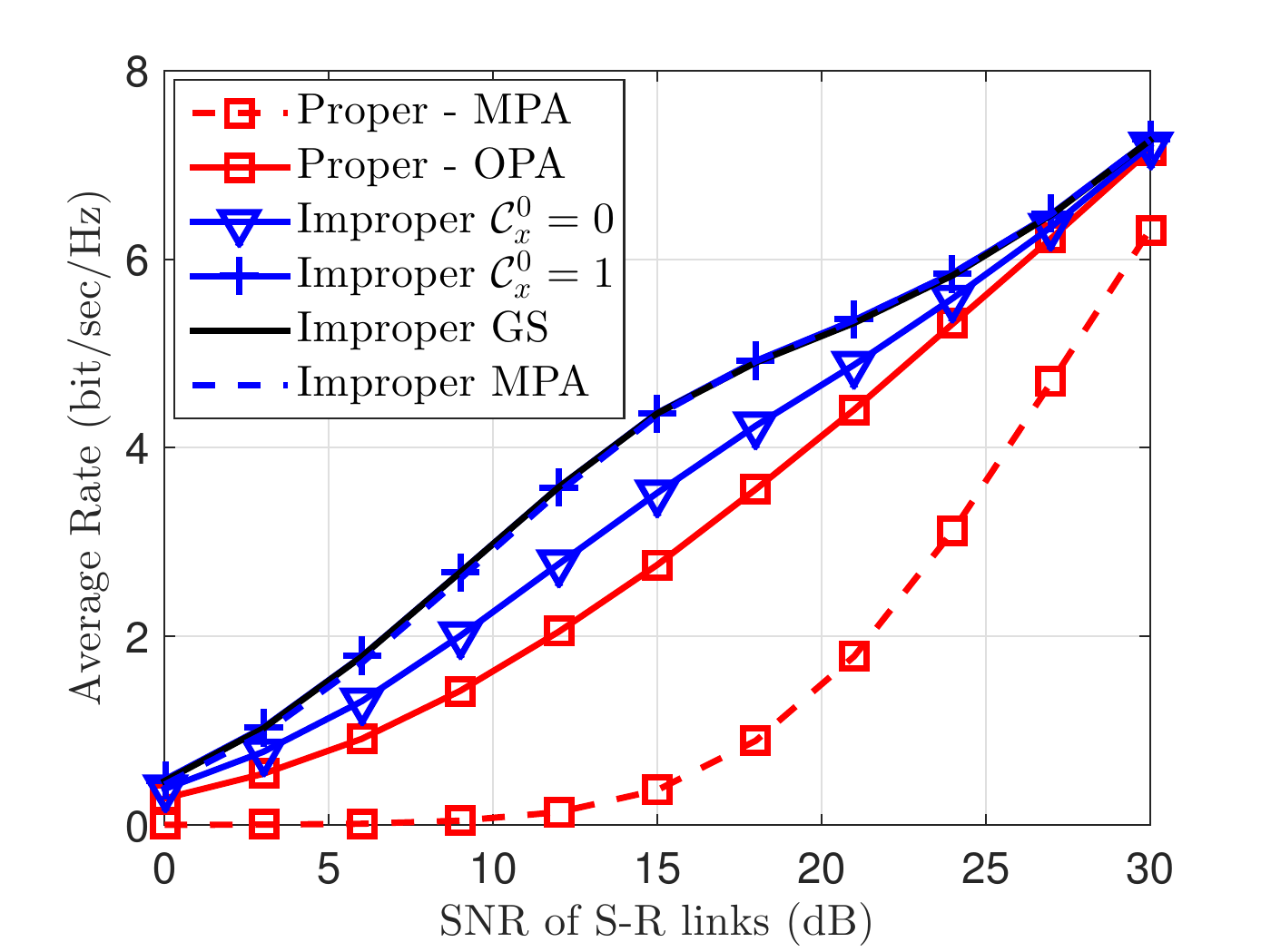}
\caption{The average achievable end-to-end rate for proper and improper signaling with different techniques versus $\gamma_{h_i}$.}
\label{fig_SR}
\end{figure}

%Secondly, we investigate the impact of the inter-relay interference channel on the system performance as shown in Fig. 3. For weak inter-relay interference channels, the interference become negligible and PGS tend to be the optimal solution. As the interference increases, all systems suffer from rates degradation but with different ways. The IGS based system makes use of the signal asymmetry property to relieve the interference impact on the relays. The benefit of the IGS can be observed to start from  $\pi_{f}>  \pi_{h}$. As for the CD solution, it achieves close to optimal performance as can be observed in Fig. 3.

%Finally, we explore the relay(s) location impact on the achievable rate. Both relays are assumed to take the same distance from each node. Fig. 4 plots the achievable rate versus the relative relays location to the source. For the PGS system, the systems suffer from rate degradation as the relays are getting far from the source because the main system data are transferred from For near to the source relays location, the links between the source and the relays are strong, while the links between  For the PGS scenario,  

%%%%%%%%%%%%%%%%%%%%%%%%%%%%%%%%%%%%%%%%%%%%%%%
%%%%%%%%%%%%%%%%%%         Conclusion       %%%%%%%%%%%%%%%%%%%
%%%%%%%%%%%%%%%%%%%%%%%%%%%%%%%%%%%%%%%%%%%%%%%
\section{Conclusion}\label{sec:conc}
In this paper, we propose to use improper signaling  in order to mitigate the inter-relay interference (IRI) in two-path relaying systems. First, we formulate an optimization problem to tune the relays transmit power and the circularity coefficient, a measure of the degree of asymmetry of the signal, to maximize the total end-to-end achievable rate of the two-path relaying system considering a power budget. We first introduce the optimal allocation of the relays power at a fixed circularity coefficient to maximize the achievable rate, then we optimize the circularity coefficient at a fixed relays power. After that we numerically optimize  the relays power and circularity coefficient jointly through a coordinate descent based method. The numerical results show a significant improvement of the total rate when the relays transmit improper signals, specifically, at mid and high IRI values. More generally, the merits of using improper signaling become significant when the first hop is the bottleneck of the system due to either week gains or the excess of IRI. 
%
%Future research lines include; considering the general system with different transmit powers and impropriety degrees at both relays along with asymmetric time allocation for the two transmission phases.    

%%%%%%%%%%%%%%%%%%%%%%%%%%%%%%%%%%%%%%%%%%%%%%%
%%%%%%%%%%%%%%%%%         Appendices      %%%%%%%%%%%%%%%%%%%
%%%%%%%%%%%%%%%%%%%%%%%%%%%%%%%%%%%%%%%%%%%%%%%
%
%%%%%%%%%%%%%%%%%%% Appendix A %%%%%%%%%%%%%%%%%%
\appendices
\section*{Appendix}
\section*{Proof of Theorem 2}
In fact, this theorem has been proved in \cite{gaafar2016letter}, however, here we give additionally graphs of the possible configurations of the rate functions ${\mathcal{R}_{i,j}}\left( {{p_{\rm{r}}^o},{{\cal C}_x}} \right)$ in \eqref{rate_first_hop_sim} and \eqref{rate_second_hop_sim}. These graphs makes the optimization problem more visually clear for the convenience of the reader.
\begin{proof} For the first case in Fig. \ref{fig:no_int}, we have four different orientations for the minimum pair of rate functions for the two paths. The minimum pair is the two decreasing functions ${\mathcal{R}_{i,2}}\left( {p_{\rm{r}}^o},{\mathcal{C}_x}\right),\forall i$ and hence, their sum will also be decreasing and the optimal solution is $\mathcal{C}_x^*=0$. Similar argument applies if the minimum pair is the two increasing functions yielding $\mathcal{C}_x^*=1$. If the minimum pair is of opposite monotonicity, we need to compute the stationary point of their sum because if there is a maximum on $0 < \mathcal{C}_x<1$, it must occur at the stationary point calculated from \cite[Proposition 3]{gaafar2016letter}.\\
  
In the second case in Fig. \ref{fig:one_int}, the intersection point, ${{\cal C}_{{i}}}$, of the two hops rates of the $i$th path, divides the $\mathcal{C}_x$ range into two intervals. In the first interval $0 < \mathcal{C}_x\leq{{\cal C}_{{i}}}$, the minimum rate of the $i$th path is ${\mathcal{R}_{i,1}}\left({p_{\rm{r}}^o}, {\mathcal{C}_x}\right)$, and in the second interval ${{\cal C}_{{i}}} < \mathcal{C}_x\leq1$, the minimum rate of the $i$th path is ${\mathcal{R}_{i,2}}\left( {p_{\rm{r}}^o},{\mathcal{C}_x}\right)$. For the $j$th path, we have two different orientations on $0 < \mathcal{C}_x<1$, either the minimum is the first or the second hop. Hence, by a similar argument as in Case 1, the result follows directly.\\

\vspace{2pt}
Finally, in the third case in Fig. \ref{fig:two_int}, we can write the total achievable rate as 
\begin{align}\label{R_tot_IGS_Cx}
&{\mathcal{R}_{{\rm{T}}}}\left({p_{\rm{r}}^o}, {\mathcal{C}_x} \right) = \frac{1}{2} \times \\ \nonumber 
& \left\{ {\begin{array}{*{20}{c}}
{\sum\limits_{i = 1}^2 {\mathcal{R}_{i,1}}\left({p_{\rm{r}}^o}, {\mathcal{C}_x}\right), }&{\rm{if}}&{0 < \mathcal{C}_x \leq \mathcal{C}_{\pi_1}}\\
{R_{\pi_2,1}}\left( {p_{\rm{r}}^o},{\mathcal{C}_x}\right)+{R_{\pi_1,2}}\left({p_{\rm{r}}^o}, {\mathcal{C}_x}\right),&{\rm{if}}&{\mathcal{C}_{\pi_1} < \mathcal{C}_x \leq \mathcal{C}_{\pi_2} }\\
{\sum\limits_{i = 1}^2 {\mathcal{R}_{i,2}}\left({p_{\rm{r}}^o}, {\mathcal{C}_x}\right), }&{\rm{if}}&{\mathcal{C}_{\pi_2} < \mathcal{C}_x < 1 }
\end{array}} \right.\hspace{-5pt}.
\end{align}
From the definition of the total rate function in \eqref{R_tot_IGS_Cx}, it can be readily verified that the function in the first interval, i.e., ${0 < \mathcal{C}_x \leq \mathcal{C}_{\pi_1}}$, is monotonically increasing in $\mathcal{C}_x$, thus the optimal solution of in this interval is $\mathcal{C}_{\pi_1}$.  Moreover, the function in \eqref{R_tot_IGS_Cx} in the last interval, i.e., ${\mathcal{C}_{\pi_2} < \mathcal{C}_x < 1 }$, is monotonically decreasing in $\mathcal{C}_x$ and hence the optimal solution in this interval is $\mathcal{C}_{\pi_2}$. If the maximum of ${\mathcal{R}_{{\rm{T}}}}\left({p_{\rm{r}}^o}, {\mathcal{C}_x} \right) $, with respect to $\mathcal{C}_x$, is in the middle interval, it must occur at  a stationary point and this concludes the proof. 
\end{proof}

\begin{figure*}
    \centering
    \begin{subfigure}[b]{0.45\textwidth}
        \includegraphics[width=\textwidth]{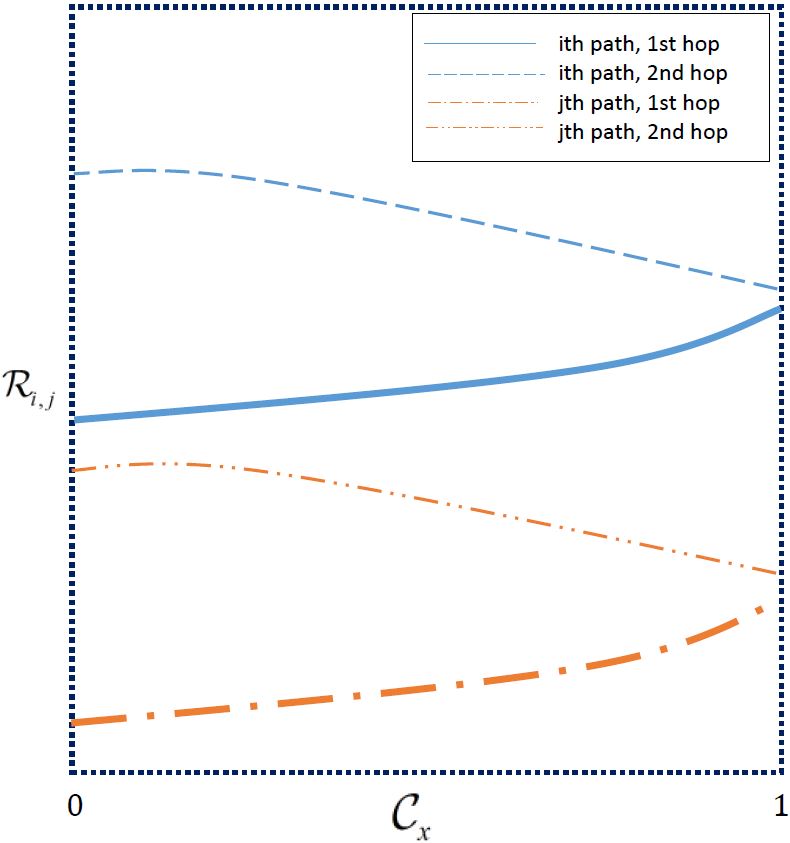}
        \caption{The minimum rate functions are both increasing. }
    \end{subfigure}
    ~ %add desired spacing between images, e. g. ~, \quad, \qquad, \hfill etc. 
      \hfill
    \begin{subfigure}[b]{0.45\textwidth}
        \includegraphics[width=\textwidth]{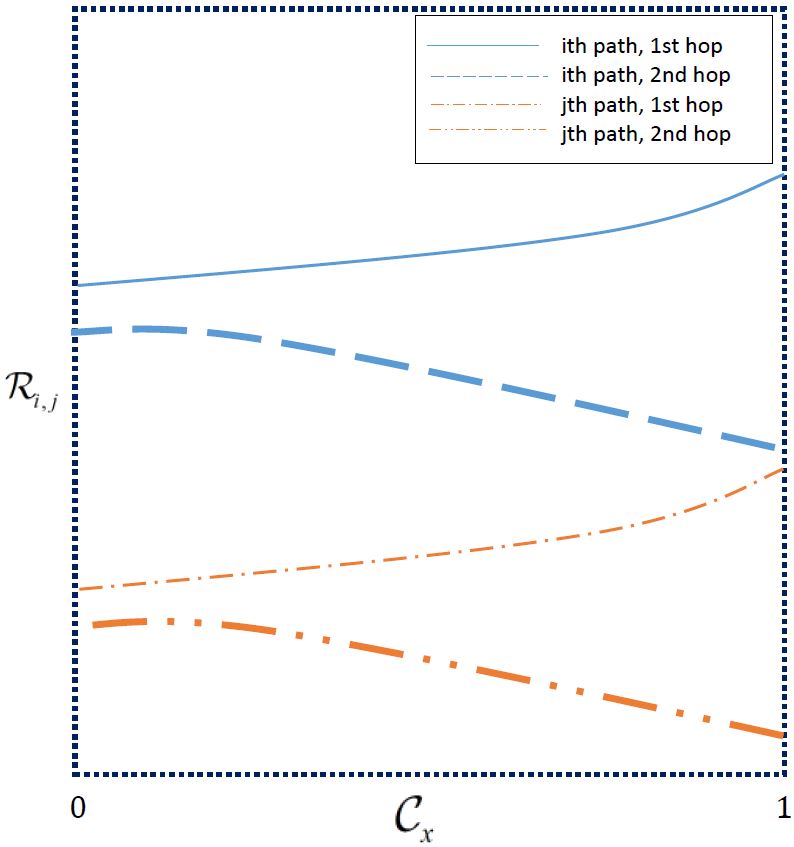}
        \caption{The minimum rate functions are both decreasing.}
    \end{subfigure}
    
    \begin{subfigure}[b]{0.45\textwidth}
        \includegraphics[width=\textwidth]{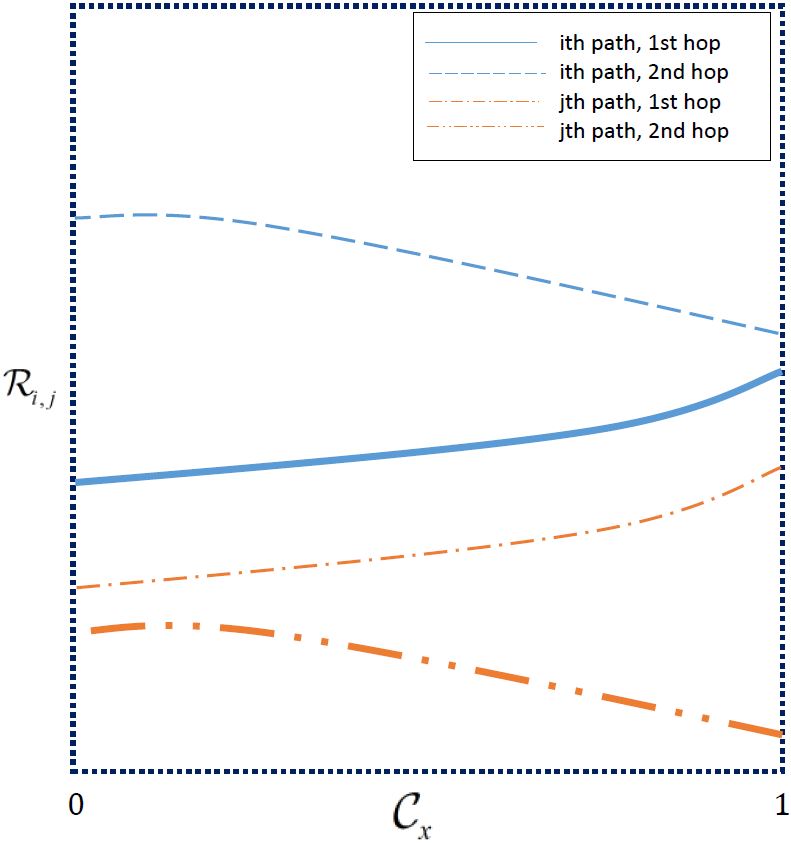}
        \caption{The minimum rate functions are increasing and decreasing.}
    \end{subfigure}
    \hfill
    \begin{subfigure}[b]{0.45\textwidth}
        \includegraphics[width=\textwidth]{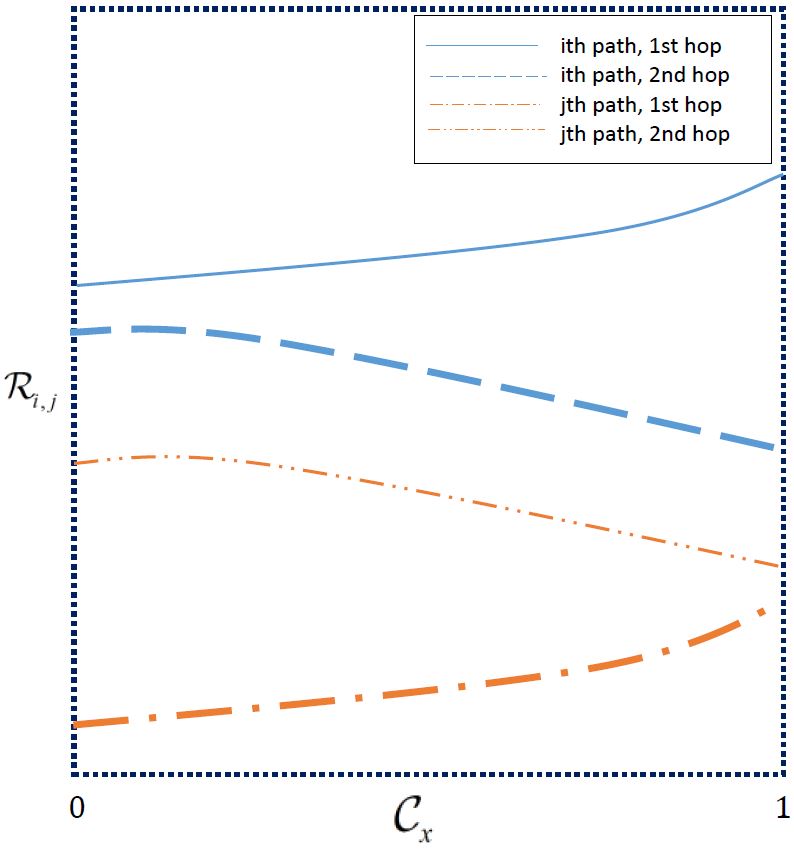}
        \caption{The minimum rate functions are decreasing and increasing.}
    \end{subfigure}
  
    \caption{Possibilities for the rate functions configurations in case of no intersections between the 1st and 2nd hops of both paths (solid lines for the minumum rate function). }\label{fig:no_int}
\end{figure*}

%===============================================================================================================================

\begin{figure*}
    \centering
    \begin{subfigure}[b]{0.45\textwidth}
        \includegraphics[width=\textwidth]{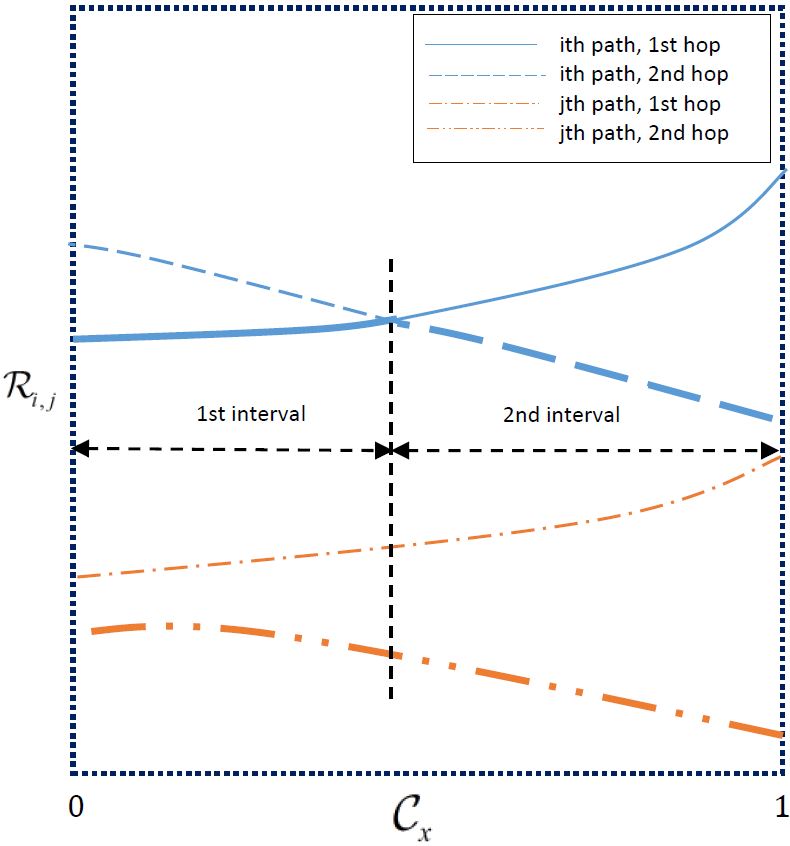}
        \caption{The minimum rate function of the $j$th path is decreasing. }
    \end{subfigure}
    ~ %add desired spacing between images, e. g. ~, \quad, \qquad, \hfill etc. 
      \hfill
    \begin{subfigure}[b]{0.45\textwidth}
        \includegraphics[width=\textwidth]{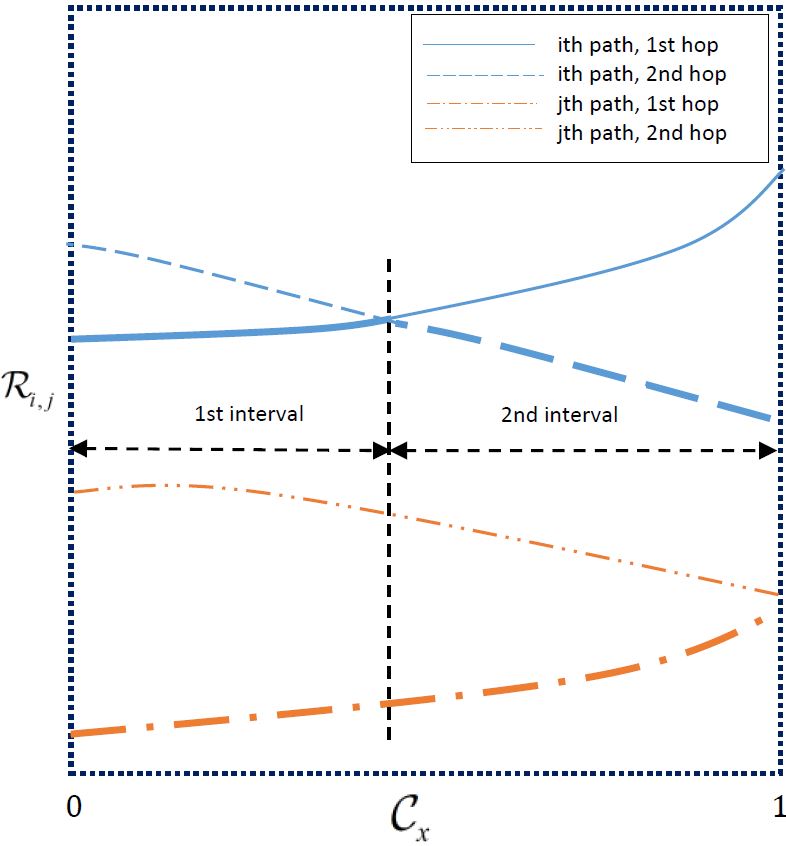}
        \caption{The minimum rate function of the $j$th path is increasing.}
    \end{subfigure}

    \caption{Possibilities for the rate functions configurations in case of existence of intersection between the 1st and 2nd hops of only one of the paths (solid lines for the minumum rate function).}\label{fig:one_int}
\end{figure*}

%===============================================================================================================================
\begin{figure*}
    \centering
        \includegraphics[width=0.45\textwidth]{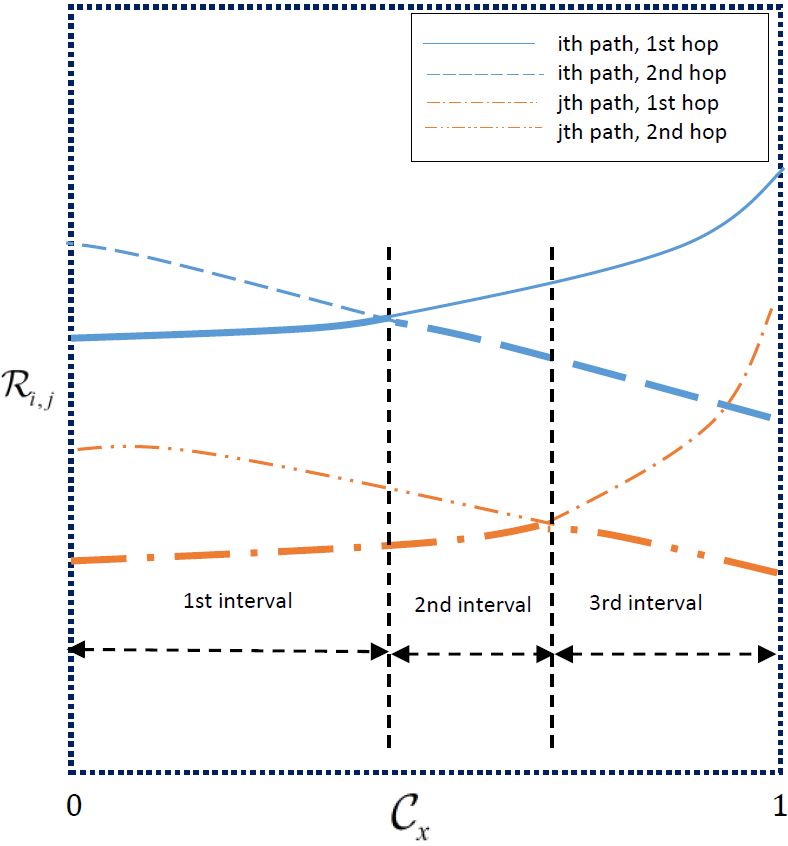}
        \caption{Possibilities for the rate functions configurations in case of existence of intersection between the 1st and 2nd hops of both paths (solid lines for the minumum rate function).}
        \label{fig:two_int}
    \end{figure*}
\clearpage
    
    \bibliographystyle{IEEEtran}

\bibliography{IEEEabrv,gaafar_ref_ICCW'17}
   
\end{document}